\pdfoutput=1                            
\def\lcversion{}					    

\documentclass[nonote,noversion,squeezed]{cdcarticle}

\def\MODE{2}

\usepackage{standalone}
\usepackage{graphicx}
\usepackage{xcolor}
\graphicspath{{./graphics/}}

\usepackage{cite}                   
\usepackage{math}                   
\usepackage[utf8]{inputenc}			
\usepackage[english]{babel}			
\usepackage[hidelinks]{hyperref}    
\usepackage[T1]{fontenc}

\usepackage{enumitem}               
\usepackage{bbm}                    
\usepackage{comment}

\newcommand{\1}{\mathbf{1}}

\newcommand{\Pp}{\mathcal{P}}
\newcommand{\U}{\mathcal{U}}
\newcommand{\V}{\mathcal{V}}
\newcommand{\K}{\mathcal{K}}
\newcommand{\W}{\mathcal{W}}	
\newcommand{\G}{\mathcal{G}}
\newcommand{\Q}{\mathcal{Q}}

\newcommand{\squeezemat}[1]{\addtolength{\arraycolsep}{-#1}}

\newcommand\blfootnote[1]{%
	\begingroup
	\renewcommand\thefootnote{}\footnote{#1}%
	\addtocounter{footnote}{-1}%
	\endgroup
}

\usepackage{subcaption}
\usepackage[margin=3mm,font=small,labelfont=bf]{caption}

\begin{document}

\title{Agent-level optimal LQG control of dynamically decoupled systems with processing delays}

\if\MODE1
\author{Mruganka Kashyap \and Laurent Lessard}
\else
\author{Mruganka Kashyap \and Laurent Lessard}\fi

\note{Submitted to CDC 2019}
\maketitle



\begin{abstract}
We consider the problem of controlling a set of dynamically decoupled plants where the plants' subcontrollers communicate with each other according to a fixed and known network topology. We assume the communication to be instantaneous but there is a fixed processing delay associated with incoming transmissions. We provide explicit closed-form expressions for the optimal decentralized controller under these communication constraints and using standard LQG assumptions for the plants and cost function. Although this problem is convex, it is challenging due to the irrationality of continuous-time delays and the decentralized information-sharing pattern. We show that the optimal subcontrollers each have an observer--regulator architecture containing LTI and FIR blocks and we characterize the signals that subcontrollers should transmit to each other across the network.
\end{abstract}

\if\MODE1\else

\blfootnote{M.~Kashyap and L.~Lessard were both with the University of Wisconsin--Madison, Madison, WI~53706, USA at the time of initial submission. M.~Kashyap is now with the Department of Electrical and Computer Engineering, and L.~Lessard is now with the Department of Mechanical and Industrial Engineering, both at Northeastern University, Boston, MA~02115, USA. \texttt{\{kashyap.mru,l.lessard\}@northeastern.edu}\\[1mm]
	This material is based upon work supported by the National Science Foundation under Grant No.~1710892.}
\fi


\vspace{-5mm}

\section{Introduction}\label{sec:intro}

When transmitting information across a network, latency can be caused either by \emph{propagation delays}, which are due to the transmission medium and are proportional to the distance the signal must travel, or by \emph{processing delays}, which are due to encoding, decoding, buffering, filtering, or other signal processing that must happen on either end of the transmission.

In scenarios where distances are relatively short, such as swarms of unmanned aerial vehicles (UAVs) communicating over a wireless network, propagation delays are negligible and it is reasonable to assume that latency is due entirely to processing delays. With UAVs, processing delays are often fixed and known, since they are dictated by the hardware capabilities of the UAVs and the bandwidth of the communication channel.

We consider the problem of controlling a set of $N$ \emph{dynamically decoupled} plants, which we refer to as \emph{agents}. A four-agent example is depicted in Fig.~\ref{fig:cartoon}. 

\begin{figure}[ht]
	\centering
	\includegraphics{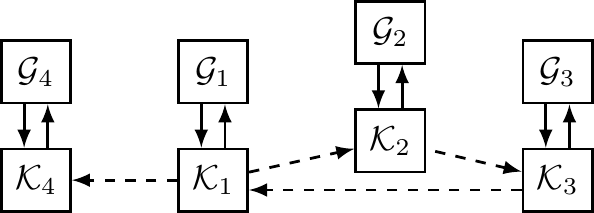}
    \caption{Example of a dynamically decoupled system. The plants $\G_i$ have corresponding controllers $\K_i$ that share information along edges of a directed graph. Implied links are not shown, e.g., $\K_2\to\K_1$ and $\K_3\to\K_4$.}
    \vspace{-1mm}
	\label{fig:cartoon}
\end{figure}

We make standard linear-quadratic-Gaussian (LQG) assumptions. That is, (1) the $\G_i$ and $\K_i$ are assumed to be continuous-time linear time-invariant (LTI) systems, (2) the exogenous noise is assumed to be Gaussian and uncorrelated between agents, and (3) the objective function is quadratic in the state and inputs.
Let $x_i$ denote the state of the agent $i$, and aggregate the agents' states into a global state vector $x$. Proceed similarly for the inputs $u$, measurements $y$, and exogenous noise $w$. We can represent the global dynamics compactly as in Fig.~\ref{fig:centralized_LQG}.
 
\begin{figure}[ht]
	\centering
	\includegraphics{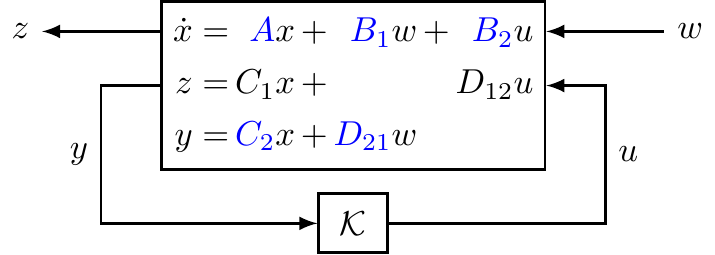}
	\caption{Classical four-block LQG problem with state-space representation of the plant. The goal is to design $\K$ to minimize the $\Htwo$ norm of the map from $w$ to $z$. Blue matrices a have block-diagonal structure due to the dynamically decoupled assumptions.}
	\label{fig:centralized_LQG}
\end{figure}

Due to our dynamically decoupled assumption, the aggregated plant matrices $A$, $B_1$, $B_2$, $C_2$, $D_{21}$ are block-diagonal. We allow the quadratic cost to couple the states and inputs of various agents, so the cost matrices $C_1$ and $D_{12}$ have no special structure.

We assume inter-controller communication happens instantaneously across the directed graph that represents the communication network (dashed lines in Fig.~\ref{fig:cartoon}), but there is a \emph{processing delay} of $\tau$ seconds for all incoming transmissions. We define the connectivity matrix $S \in \{0,1\}^{N\times N}$ as $S_{ij}=1$ if there is a directed path from controller $j$ to controller $i$ and $S_{ij}=0$ otherwise. The processing delays impose block structural constraints on $\K$. Specifically, $\K_{ij}$ is: zero if $S_{ij}=0$ (no communication), delayed by $\tau$ if $S_{ij}=1$ and $i\ne j$ (processing delay), and unconstrained if $S_{ij}=1$ and $i=j$ (no delay). For example, the connectivity and controller structures for the example of Fig.~\ref{fig:cartoon} are
\[
S\!=\!\squeezemat{1mm}\bmat{1&1&1&0\\ 1&1&1&0\\ 1&1&1&0\\ 1&1&1&1}\!,\,\,
\K \!\sim\! \left[\begin{array}{cccc}\K_{11}&e^{-s\tau}\K_{12}&e^{-s\tau}\K_{13}&0\\ e^{-s\tau}\K_{21}&\K_{22}&e^{-s\tau}\K_{23}&0\\ e^{-s\tau}\K_{31}&e^{-s\tau}\K_{32}&\K_{33}&0\\ e^{-s\tau}\K_{41}&e^{-s\tau}\K_{42}&e^{-s\tau}\K_{43}&\K_{44}\end{array}\right]\!.
\]
In this example, the controller $\K_3$ receives (delayed) measurements from $\G_1$ and $\G_2$, instantaneous measurement from $\G_3$, and no information from $\G_4$. In other words,
\[
u_3(t) = f_3\bigl( y_1(0:t-\tau), y_2(0:t-\tau), y_3(0:t) \bigr)
\]
for some function $f_3$. In the remainder of the paper, we derive the optimal structured $\K$ as a function of the connectivity $S$, delay $\tau$, and state-space matrices of the four-block plant (Fig.~\ref{fig:centralized_LQG}). We also provide a detailed and intuitive \emph{agent-level} description of the optimal $\K_i$.

\paragraph{Summary of relevant literature.}

The sparsity and delay constraints on our plant and controller are an example of a \textit{quadratically invariant} (QI) system~\cite{rotkowitz2005characterization,rotkowitz2010convexity}. Therefore, the problem of finding the optimal structured linear controller amounts to solving a convex (but infinite-dimensional) optimization problem.

When the structure does not contain delays, the exact optimal controller may be computed using vectorization~\cite{rotkowitz_vectorization,vamsi_elia}. In general, vectorization does not produce a minimal realization, nor does it elucidate structure, such as a controller-estimator separation or an interpretation of the signals communicated between subcontrollers.
Explicit solutions to special cases of output-feedback have been reported, such as: triangular~\cite{lessard2015optimal,tanaka2014optimal}, broadcast~\cite{lessard2012decentralized}, or dynamically decoupled~\cite{kim2015explicit,kashyap2019explicit} cases.

When delays are present, we distinguish the cases of discrete vs.\ continuous-time delay. In discrete time, the delay operator $z^{-1}$ is rational, so delays may be absorbed into the plant and the problem reduces to the non-delayed case~\cite{spdel}. In continuous time, this reduction is not possible because the delay operator $e^{-s\tau}$ is now irrational. One approach is to use a Pad\'{e} approximation for the delay~\cite{yan1996teleoperation} followed by vectorization. Alternatively, a Galerkin-style finite-dimensional approximation~\cite{scherer02,voulgaris_stabilization} can be used.

The present work seeks an explicit solution for the delayed case that provides structure and intuition without resorting to approximations or vectorization. 
Some special cases have been addressed in the literature. For centralized problems with a fixed loop delay (\emph{dead time}), a loop-shifting technique involving FIR blocks can transform the problem into an equivalent LQG problem with a finite-dimensional LTI plant~\cite{mirkin2003every,mirkin2003extraction}. A similar idea was used in the discrete-time case to decompose the structure into dead time and FIR components, which can be optimized separately~\cite{lampHtwoDelay}.

The aforementioned loop-shifting approach was also extended to the \emph{adobe delay} case, where the feedback loop contains both a delayed and a non-delayed path~\cite{mirkin2011dead,mirkin2012h2}. This approach was used to obtain explicit optimal controllers for bilateral teleoperation problems, where the controllers communicate across a delayed channel~\cite{kristalny2012decentralized,cho2012h2}. The idea was also generalized to haptic interfaces that have two-way communication with a shared virtual environment~\cite{kristalny2013decentralized}. These are special cases of the general problem we will treat, described in Section~\ref{sec:intro}, where $S$ is the all-ones matrix.

\paragraph{Overview.}

Our work generalizes the agent-level solution for the non-delayed case reported in~\cite{kashyap2019explicit} by leveraging the loop-shifting decomposition for adobe delays from~\cite{mirkin2012h2}. The paper is organized as follows.
Section~\ref{sec:notation} covers our notation and assumptions.
Section~\ref{sec:result} presents our solution to the general problem as a function of the connectivity matrix $S$ and processing delay $\tau$. Our results are presented in continuous time with an infinite-horizon cost but can easily be generalized to discrete time and/or a finite-horizon cost. We also characterize the signals that should be transmitted between subcontrollers and show that each subcontroller has an intuitive observer-regulator structure.
Finally, Section~\ref{sec:diss} shows that we recover the expected limiting cases in the limits of zero or infinite processing delay.

\section{Preliminaries}\label{sec:notation}

\paragraph{Notation.} 
State-space notation for transfer functions:
\begin{align*}
	\G(s) = \left[\begin{array}{c|c}%
		A&B\\ \hlinet C&D\end{array}\right]= D+C(sI-A)^{-1}B.
\end{align*}
We let $N$ denote the total number of agents and $[N] \defeq \{1,\dots,N\}$. The $i^\text{th}$ subsystem has state dimension $n_i$, input dimension $m_i$, and measurement dimension $p_i$. The global state dimension is $n \defeq n_1+\cdots+n_N$ and similarly for $m$ and $p$. The identity matrix of size $k$ is denoted $I_k$. We write $\blkdiag(\{X_i\})$ to denote the block-diagonal matrix formed by the blocks $\{X_1,\dots,X_n\}$ and $\diag(X)$ to denote the block-diagonal matrix formed by the diagonal blocks of $X$.
The zeros used throughout are matrix or vector zeros and their sizes are determined from context.
The symbol $\otimes$ denotes the Kronecker product.

For $i\in [N]$, we write $\underline{i} \subseteq [N]$ to denote the \textit{descendants} of node $i$, i.e., the set of nodes $j$ such that there is a directed path from $i$ to $j$. Likewise, $\bar{i} \subseteq [N]$ denotes the \textit{ancestors} of node $i$. Similarly, $\bar{\bar{i}}$ and $\underline{\underline{i}}$ denote the \textit{strict ancestors} and \textit{strict descendants}, respectively. We also write $s^c = [N]\setminus s$ to denote the complement.
For example, in the graph of Fig.~\ref{fig:cartoon}, we have $\underline{2} = \{1,2,3,4\}$, $\underline{\underline{2}} = \{1,3,4\}$, and $\bar{\bar{2}} = \{1,3\}$.
We also use this notation to index matrices.
For example, if $X$ is a $4\times 4$ block matrix associated with Fig.~\ref{fig:cartoon}, then
$X_{1\bar{\bar{2}}} = \bmat{X_{11} & X_{13}}$.

We will require the use of specific partitions of the identity matrix. We define $I_{n} \defeq \blkdiag(\{I_{n_i}\})$ and for each agent $i \in [N]$, we define $E_{n_i} \defeq (I_{n})_{:i}$ (the $i^\textup{th}$ block column of $I_{n}$). Similar to the descendant and ancestor definitions, $n_{\underline{i}}\defeq\sum_{k\in \underline{i}}n_k$ and $n_{\bar{i}}\defeq\sum_{k\in \bar{i}}n_k$. The dimensions of $E_{{n_{\bar{i}}}}$ and $E_{{n_{\underline{i}}}}$ are determined by the context of use. Finally, $1_{n}$ is the $n\times 1$ matrix of $1$'s.

\paragraph{Processing delay notation~\cite{mirkin2012h2}.} 
The \textit{adobe delay} matrix $\Lambda_m^i \defeq \blkdiag(I_{m_i},e^{-s\tau}I_{m_{\underline{\underline{i}}}})$ leaves block $i$ unchanged and imposes a delay of $\tau$ on all strict descendants of $i$.
The completion operator $\pi_\tau\{\cdot\}$ acts on a state-space system delayed by $\tau$ and returns the FIR system with support on $[0,\tau]$ that completes it:
\[
\squeezemat{2pt}
\pi_\tau\Biggl\{\left[\begin{array}{c|c}A&B\\ \hlinet C&D\end{array}\right]\!e^{-s\tau}\Biggr\}
\defeq\left[\begin{array}{c|c}A&e^{-A\tau}\!B\\ \hlinet C&D\end{array}\right]-\left[\begin{array}{c|c}A&B\\ \hlinet C&D\end{array}\right]e^{-s\tau}\!.
\]
We also define the function $\Gamma : (\Omega,\Lambda_m^i) \mapsto (\tilde{\Omega},\Pi_u,\Pi_b)$, which maps a four-block plant $\Omega$ defined as
\begin{equation}\label{eq:fourblock}
	\Omega \defeq
	\bmat{ \U & \V \\ \W & \G } \defeq
	\left[\begin{array}{c|cc}
		A & B_1 & B_2\\ \hlinet
		C_1 & 0 & D_{12} \\
		C_2 & D_{21} & 0 \end{array}\right]
\end{equation}
and adobe delay matrix $\Lambda_m^i$ to a modified plant $\tilde{\Omega}$ and FIR systems $\Pi_u$ and $\Pi_b$. See Appendix~\ref{sec:appendix_gamma} for details.

\paragraph{Problem statement.}
Given a four-block plant~\eqref{eq:fourblock} and associated connectivity matrix $S$ as described in Section~\ref{sec:intro} (refer to Fig.~\ref{fig:centralized_LQG}), the plants are dynamically decoupled but the cost function may couple the states and inputs of the different agents. Iit follows that $\W$ and $\G$ are block diagonal but $\U$ and $\V$ need not. Let $\mathbb{S}_\tau$ be the set of causal structured LTI controllers with processing delay $\tau$, as described in Section~\ref{sec:intro}.
The problem is to find $\K$ to 
\begin{equation}\label{opt}
\begin{aligned}
\underset{\K}{\minimize}
\qquad & \normm{\,\U + \V\K(I-\ \G\K)^{-1}\W\,}_2^2 \\
\subject \qquad & \K \in \mathbb{S}_{\tau}
\text{ and $\K$ stabilizes $\Omega$.}
\end{aligned}
\end{equation}
We now describe and explain the technical assumptions we make on the four-block plant $\Omega$ and connectivity $S$.

\paragraph{Riccati assumptions.} Four matrices $(A,B,C,D)$ are said to satisfy the \emph{Riccati assumptions}~\cite{kim2015explicit,mirkin2012h2} if:
\begin{enumerate}[label*=R\arabic*., ref=R\arabic*]
	\item $C^\tp D=0$ and $D^\tp D>0$. \label{harf}
	\item $(A,B)$ is stabilizable.
	\item $\begin{bmatrix}A-j\omega I&B\\C&D\end{bmatrix}$ is full column rank for all $\omega \in \R$.
\end{enumerate}
If the Riccati assumptions hold, there is a unique stabilizing solution to the associated algebraic Riccati equation, which we write as $(X,F)=\ric(A,B,C,D)$. Thus $X \succeq 0$ satisfies
$
A^\tp X+XA+C^\tp C -XB(D^\tp D)^{-1}B^\tp X=0
$ with $A+BF$ Hurwitz, where $F\defeq -(D^\tp D)^{-1}B^\tp X$.

\begin{assumption}[\textbf{System assumptions}]
	\label{Ass:System}
	For the $N$ interacting agents, we will assume the following.
	\begin{enumerate}[label*=\arabic{thm}.\arabic*., ref=\arabic{thm}.\arabic*]
		\item $A_{ii}$ is Hurwitz for all $i\in[N]$, i.e., $A$ is Hurwitz. \label{Ass:System_hurwitz}
		\item The Riccati assumptions hold for $(A,B_2,C_1,D_{12})$ and for $(A_{ii}^\tp ,C_{2_{ii}}^\tp ,B_{1_{ii}}^\tp ,D_{{21}_{ii}}^\tp )$ for all $i\in[N]$. \label{Ass:System_riccati}
		\item $D_{12_{:\underline{i}}}^\tp D_{12_{:\underline{i}}}=I$ and $D_{21_{ii}}D_{21_{ii}}^\tp=I$ for all $i\in[N]$. \label{Ass:System_D1221}
	\end{enumerate}
\end{assumption}
\noindent Assumption~\ref{Ass:System_hurwitz} is an assumption of nominal stability, carried over from Kim et~al.~\cite{kim2015explicit}. Assumption~\ref{Ass:System_riccati} ensures the necessary condition that the centralized LQR problem and the individual agents' estimation problems are non-degenerate. Assumption~\ref{Ass:System_D1221} simplifies the exposition of the delayed problem~\cite{mirkin2011dead,mirkin2012h2}, though the results still hold for the general case~\cite[Rem.~3.2]{mirkin2011dead}.

\section{Main Result}\label{sec:result}

Our main result is an agent-level description of the optimal controller that solves~\eqref{opt}. That is, we provide explicit state-space formulas for each $\K_i$ and describe which signals should be transmitted and received between agents.

\begin{thm}
	\label{thm:3}
	Consider a general instance of the structured optimal control problem described in Section~\ref{sec:intro} and suppose Assumption~\ref{Ass:System} holds. Consider the four-block sub-plant for agent $i$ and its descendants:
	\begin{align*}
			{\Omega}_i \defeq
			\left[\begin{array}{c c}
				{\U}_{:i} & {\V}_{:\underline{i}}\\
				{\W}_{ii} & {\G}_{i\underline{i}}\end{array}\right] \defeq
			\left[\begin{array}{c|cc}
				A_{\underline{ii}} & B_{1_{\underline{i}i}} & B_{2_{\underline{ii}}} \\[2pt] \hlinet
				{C}_{1_{:\underline{i}}} & 0 & D_{12_{:\underline{i}}} \\
				C_{2_{i\underline{i}}} & D_{21_{ii}} & 0\end{array}\right].
	\end{align*}
	Now apply
	$(\tilde{\Omega}_i,\Pi_{u_i},\Pi_{b_i})=\Gamma(\Omega_i,\Lambda_m^i)$. Define the estimation gain $\tilde{F}^i$ and control gain $L^i$ by
	\begin{subequations}\label{eq:Riccati}
		\begin{align}
		(\tilde{X}^i,\tilde{F}^i) &\defeq \ric\bigl(A_{\underline{ii}},\tilde{B}_{2_{\underline{ii}}},\tilde{C}_{1_{:\underline{i}}},D_{12_{:\underline{i}}}\bigr),\\
		(Y^i,{L^i}^\tp) &\defeq
		\ric(A_{ii}^\tp,C_{2_{ii}}^\tp,B_{1_{ii}}^\tp,D_{{21}_{ii}}^\tp).
		\end{align}
    \end{subequations}
	A realization of the optimal controller $\K_i$ for agent $i$ that solves~\eqref{opt} is given by the state-space equations
\begin{subequations}\label{eq:state_agent_Ktau_reduced}
    \begin{align}
        \dot \eta_{i,\underline{i}} &= (A_{\underline{ii}}+(LC)^i)\eta_{i,\underline{i}} + (LC)^i \sum_{k\in \bar{\bar i}}\eta_{k,\underline{i}}(t-\tau) \notag\\
        &  \qquad\qquad + \tilde{B}_{2_{\underline{ii}}} \tilde{\nu}_{i,\underline{i}}- E_{n_{\underline{i}}}^{\tp} E_{n_i}L^i (y_i+\Pi_{b_i} \nu_{i,\underline{i}}),\\
        \tilde{\nu}_{i,\underline{i}} &= \tilde{F}^i \eta_{i,\underline{i}}, \\
        u_i &= E_{m_i}^\tp E_{m_{\underline{i}}}\biggl( \nu_{i,\underline{i}}+ \sum_{k\in\bar{\bar{i}}}\tilde{\nu}_{k,\underline{i}}(t-\tau) \biggr),
    \end{align}
\end{subequations} 
where $(LC)^i\defeq E_{n_{\underline{i}}}^{\tp} E_{n_i} L^i C_{2_{ii}} E_{n_{i}}^{\tp} E_{n_{\underline{i}}}$ and $\nu_{i,\underline{i}} \defeq \Pi_{u_i} \tilde{\nu}_{i,\underline{i}}$.
\end{thm}

\begin{proof}
	See Appendix \ref{App:A}.
\end{proof}

We use a slight abuse of notation in~\eqref{eq:state_agent_Ktau_reduced}. These are time-domain equations but $\Pi_{b_i}$ and $\Pi_{u_i}$ are FIR transfer matrices. Products such as $\Pi_{b_i} \nu_{i,\underline{i}}$ are to be interpreted as the time-domain output of $\Pi_{b_i}(s)$ with input $\nu_{i,\underline{i}}(t)$.

Fig.~\ref{fig:Ki} shows the signals received and transmitted by each subcontroller. Fig.~\ref{fig:local} shows a more detailed block diagram. The optimal controller is an interconnection of state-space systems, FIR blocks, and delays.


\section{Limiting cases}\label{sec:diss}

\begin{figure}[ht]
	\centering
	\includegraphics{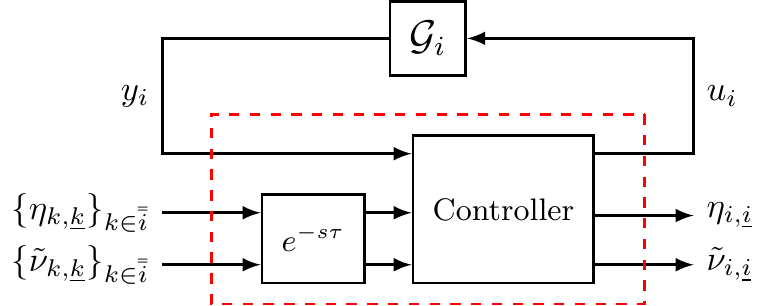}
	\caption{Block diagram representation of a local plant $\G_i$ and associated controller $\K_i$ (dashed red box). Controller $i$ receives local measurements $y_i$ and information $(\eta,\tilde\nu)$ from its strict ancestors $\bar{\bar{i}}$ (which incurs a processing delay $e^{-s\tau}$). The controller then computes  and transmits information to its strict descendants $\underline{\underline{i}}$.}
	\label{fig:Ki}
\end{figure}

\begin{figure*}[htb]
	\centering
	\includegraphics{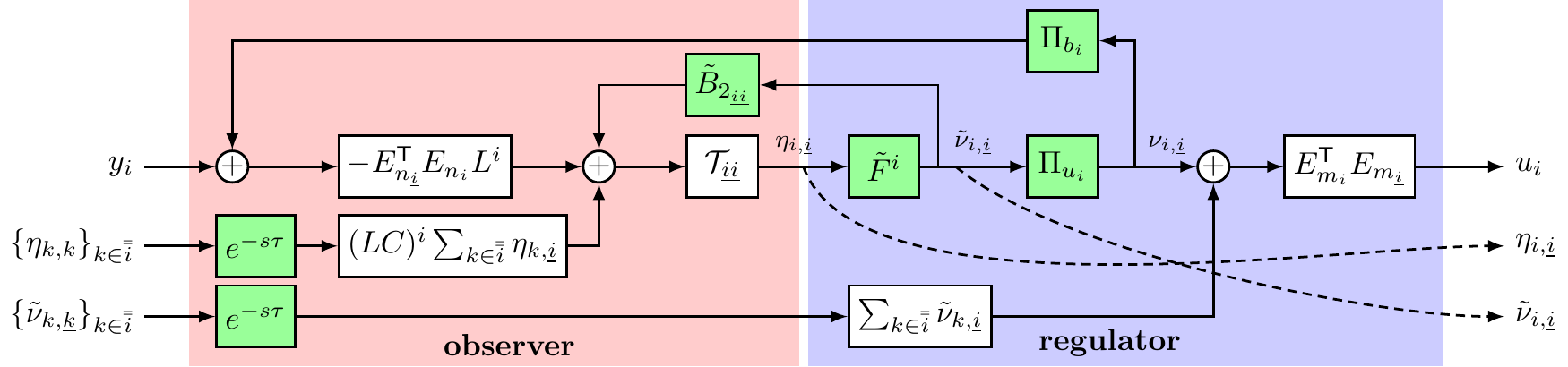}
	\caption{Expanded diagram of the optimal controller shown in Fig.~\ref{fig:Ki} that reveals the observer-regulator structure of $\K_i$ from Theorem~\ref{thm:3}. The green blocks represent blocks that depend on the processing delay $\tau$.  
		Agent $i$ updates its local innovation $\eta_{i,\underline i}$ and computes its local input $u_i$ using the local measurement $y_i$ and the delayed innovations $\eta_{k,\underline k}$ and partial inputs $\tilde v_{k,\underline k}$ from  strict ancestors $k\in \bar{\bar{i}}$.
		For simplicity, we used the following notation for the observer dynamics: $\mathcal{T}_{\underline{ii}} \defeq  (sI-A_{\underline{ii}}-(LC)^i)^{-1}$, where $(LC)^i\defeq E_{n_{\underline{i}}}^{\tp} E_{n_i} L^i C_{2_{ii}} E_{n_{i}}^{\tp} E_{n_{\underline{i}}}$.}
	\label{fig:local} 
\end{figure*}

We now study the behavior of our optimal controller as the processing delay $\tau$ varies. Setting $\tau=0$, nodes belonging to the same connected component can be treated as a single node due to the instantaneous communication assumption. For example, if $\tau=0$ in Fig.~\ref{fig:cartoon}, nodes $1$, $2$, and $3$ become a single node, and we are left with the equivalent two-node graph $S = \left(\begin{smallmatrix}1 & 0 \\ 1 & 1 \end{smallmatrix}\right)$. This special case was solved in~\cite{kashyap2019explicit} and we recover this result by setting $\tau=0$ in Theorem~\ref{thm:3}. At the other extreme, the case $\tau\to\infty$ leads to a controllers that do not communicate at all. These results are formalized in Corollary~\ref{corr:1} below.

\begin{cor}\label{corr:1}
Consider the setting of Theorem~\ref{thm:3}.
\begin{enumerate}[label*=\arabic*), ref=\arabic*]
    \item If there is no processing delay ($\tau=0$), The controller $\K_{0}$ that solves~\eqref{opt} is
    \begin{subequations}
		\begin{align*}
            \hspace{-2mm}\dot \eta_{i,\underline{i}} &= A_{\underline{ii}} \eta_{i,\underline{i}} + 
            (LC)^i\sum_{k \in \bar{i}}\eta_{k,\underline{i}} + {B}_{2_{\underline{ii}}} \nu_{i,\underline{i}} - E_{n_{\underline{i}}}^{\tp} E_{n_i}L^iy_i\\
            \hspace{-1mm}u_i &= E_{m_i}^\tp\sum_{j\in\bar{i}} F_j \eta_j,
        \end{align*}
    \end{subequations}
    where $\nu_{i,\underline{i}} \defeq {F}^i \eta_{i,\underline{i}}$ and $F_j\defeq E_{m_{\underline{j}}}F^{j}E_{n_{\underline{j}}}^\tp$ and the remaining parameters are defined in Theorem~\ref{thm:3}.
    
    \item If there is infinite processing delay ($\tau\to\infty$), the controller $\K_{\infty}$ that  solves~\eqref{opt} is
	\begin{subequations}\label{eq:Kinfty}
		\begin{align}
			\dot \eta_{i,i} &= A_{ii} \eta_{i,i} + L^iC_{2_{ii}}\eta_{i,i}+ B_{2_{ii}} u_i-L^iy_i\\
			u_i &= F_\infty^i\eta_{i,i},
		\end{align}
	\end{subequations}
	where 
    $(X_\infty^i,F_\infty^i) \defeq		\ric(A_{ii},B_{2_{ii}},C_{1_{:{i}}},D_{12_{:{i}}})$.
\end{enumerate}
\end{cor}

\begin{proof}
    When $\tau \to 0$, we have $\Pi_{u_i} = I$ and $\Pi_{b_i} =0$. Moreover, $\tilde\Omega_i = \Omega_i$ and $\tilde F^i = F^i$.
    When $\tau \to \infty$, we also have $\Pi_{u_i} \to I$ and $\Pi_{b_i} \to 0$ but this time, $\tilde B_{2_{\underline{ii}}} \to \diag(B_{2_{\underline{ii}}})$. The corresponding control gain $\tilde{F}^i$ reduces to an augmented nominal gain $E_{m_i}F_{\infty}^{i}E_{n_i}^\tp$ after elimination of the uncontrollable ($\underline{\underline{i}}$) and unobservable ($\underline{i}^c$) modes to obtain~(\ref{eq:Kinfty}). Letting $\tau\to\infty$ is equivalent to solving $N$ separate LQG problems (the global cost matrices are block-diagonal).
\end{proof}

The effect of changing the processing delay $\tau$ is illustrated in the block diagram of Fig.~\ref{fig:local}. When $\tau$ varies, only the green blocks are affected. The remaining parts of the controller can be designed without knowing $\tau$.

\section{Conclusion}
\label{sec:conclufuture}
In this work, we considered an optimal control problem where dynamically decoupled agents communicate over a network and incur processing delays to receive transmissions from neighboring controllers. We described the structure of the optimal control strategies at the level of individual agents (Fig.~\ref{fig:local}), which shows which signals should be transmitted and which parts of the controller depend on the processing delay.

\appendix

\section{Definition of the \texorpdfstring{$\Gamma$}{Gamma} function}\label{sec:appendix_gamma}

The input matrices $B_2$ and $D_{12}$ of $\Omega$ defined in~\eqref{eq:fourblock} are partitioned according to the blocks of $\Lambda_m^i$.
Therefore, ${B}_2=\bmat{ B_{2_0} & B_{2_{\tau}} }$, where the two blocks have delays of $0$ and $\tau$, respectively. We partition $D_{12}=\bmat{ D_{12_0} & D_{12_{\tau}} }$ similarly.
Now define the Hamiltonian
\[
\squeezemat{3pt}
H\!=\! \bmat{H_{11} & H_{12} \\ H_{21} & H_{22}} \!\defeq\!
\squeezemat{2pt} \bmat{ A - B_{2_0} D_{12_0}^\tp C_1 & -B_{2_0}B_{2_0}^\tp \\
-C_1^\tp P_\tau C_1 & -A^\tp 
+ C_1^\tp D_{12_0}  B_{2_0}^\tp}\!,
\]
where $P_0 \defeq D_{{12}_0}D_{{12}_0}^\tp$ and $P_{\tau} \defeq I-P_0$, 
and define its symplectic matrix exponential as $\Sigma \defeq e^{H\tau}$.
Define the modified matrices corresponding to $B_2$ and $C_1$ as
\begin{align*}
	\tilde{B}_2 &\defeq
	\left[\begin{array}{cc}
		B_{2_0} & \Sigma_{12}^\tp C_{1}^\tp D_{{12}_{\tau}}+\Sigma_{22}^\tp B_{2_{\tau}}\end{array}\right]\\
	\tilde{C}_1 &\defeq \left(P_{\tau} C_{1} + P_0 C_{1} \Sigma_{22}^\tp-D_{{12}_{0}} B_{2_0}^\tp \Sigma_{21}^\tp \right)\Sigma_{22}^{-\tp},	
\end{align*}
where the $\Sigma_{ij}$ are partitioned the same was as the $H_{ij}$. The modified four-block plant output by $\Gamma$ is then
\begin{equation}\label{eq:omegatilde}
	\tilde{\Omega} \defeq
	\bmat{ \tilde{\U} & \tilde{\V} \\ \W & \tilde{\G} } \defeq
	\left[\begin{array}{c|cc}
		A & B_1 & \tilde{B}_2\\ \hlinet
		\tilde{C}_1 & 0 & D_{12} \\
		C_2 & D_{21} & 0 \end{array}\right],
\end{equation}
where $\tilde{B}_2 = \bmat{ B_{2_0} & \tilde{B}_{2_{\tau}} }$ is partitioned the same way as the original $B_2$. Finally, define the FIR systems
\[
	\bmat{\tilde{\Pi}_u\\ \tilde{\Pi}_b} \defeq
	\pi_{\tau}\!
	\left\{ \left[\begin{array}{cc|c}
		H_{11} & H_{12} & B_{2_{\tau}}\\
		H_{21} & H_{22} & -C_1^\tp D_{12_{\tau}}\\ \hlinet
		D_{12_0}^\tp{C}_1 & B_{2_0}^\tp & 0 \\
		C_2 & 0 & 0\end{array}\right]e^{-s\tau}     
	\right\}
\]
and $\Gamma$ outputs 
$\Pi_{u} \defeq {\left[\begin{array}{cc}I& \tilde{\Pi}_u\\
0&I\end{array}\right]}$ and $\Pi_b \defeq \left[\begin{array}{cc}0& \tilde{\Pi}_b \end{array}\right]$.

\section{Proof of Theorem~\ref{thm:3}}\label{App:A}
\begin{figure*}[htb]
	\centering
	\includegraphics{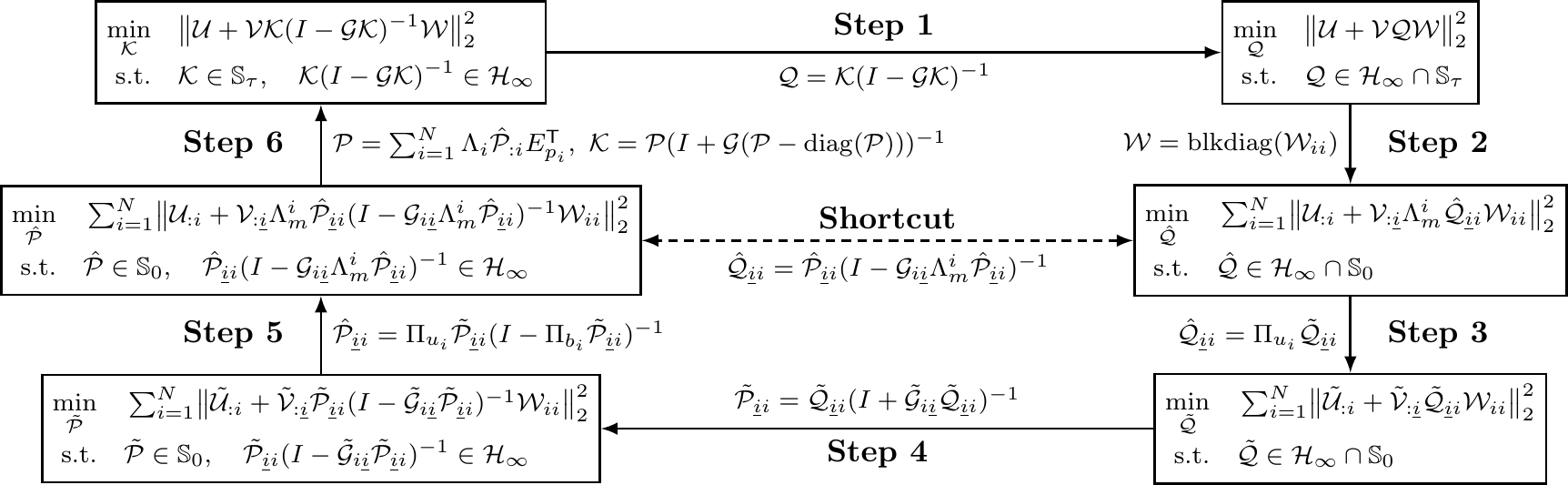}
	\caption{Sequence of transformations used in the proof of Theorem~\ref{thm:3}. The proof starts in the upper-left box, which is our original structured optimization problem~\eqref{opt} and follows the arrows clockwise.}
	\label{fig:proof} 
\end{figure*}
A roadmap for the proof is shown in Fig.~\ref{fig:proof}. We begin in the top left block, which is~\eqref{opt}, then we follow the arrows, each of which is explained in the following paragraphs.

\noindent\textbf{Step 1.} Assumption~\ref{Ass:System_hurwitz} together with quadratic invariance~\cite{rotkowitz2005characterization,rotkowitz2010convexity} ensures that all stabilizing controllers are parameterized by $\Q=\K(I-\G\K)^{-1}$, where the Youla parameter $\Q$ is stable and has the same structure as $\K$. This leads to a convex model-matching problem in $\Q$.

\vspace{1mm}
\noindent\textbf{Step 2.} Note that $\W$ is block-diagonal and $\Q_{\underline{i}i}=\Lambda_m^i\hat\Q_{\underline{i}i}$, where $\Lambda_m^i$ is an adobe delay matrix (defined in Section~\ref{sec:notation}), and $\hat Q$ is delay-free. Thus, we may separate the cost by its block columns and optimize each summand separately. A similar separation technique was leveraged in~\cite{kashyap2019explicit,kim2015explicit} to solve the non-delayed version of this problem.

\vspace{1mm}
\noindent\textbf{Step 3.} We now apply the loop-shifting result of Mirkin et al.~\cite[Thm.~1]{mirkin2012h2}, which we state below as Lemma~\ref{lem:1}. This result transforms an optimal control problem containing an adobe delay $\Lambda_m^i$ to one with no delays. The original controller can then be recovered via a transformation involving FIR blocks.

\begin{lem}[{\!\!\cite[Thm.~1]{mirkin2012h2}}]\label{lem:1}
    Consider the following structured optimal control problem subject to Assumption~\ref{Ass:System} similar to~\eqref{opt}, but with a simpler structure: there are two control inputs, and the second one is delayed by $\tau$. In other words, we would like to solve
    \begin{equation*}\label{opt_mirkin}
        \underset{\Lambda_m \K \; \stabilizes \; \Omega}{\minimize} \quad
        \normm{\,\U + \V\Lambda_m\K(I-\ \G\Lambda_m\K)^{-1}\W\,}_2^2
    \end{equation*}
    where $\Lambda_m \defeq \blkdiag(I, e^{-s\tau} I)$. 
    Let $(\tilde{\Omega},\Pi_u,\Pi_b)=\Gamma(\Omega,\Lambda_m)$, where $\Gamma$ is defined in Section~\ref{sec:notation} and Appendix~\ref{sec:appendix_gamma}.
    Then, $\Lambda_m\K$ stabilizes $\Omega$ if and only if  $\tilde{K}$ stabilizes $\tilde{\Omega}$, where $\K$ and $\tilde\K$ are related via the bijective transformation
    $
        \K = \Pi_{u}\tilde{\K}(I-\Pi_b\tilde{\K})^{-1}
    $.
\end{lem}
\noindent Transfer matrices $\tilde\U$, $\tilde\V$, $\W$ are rational with realizations given in~\eqref{eq:omegatilde}, so we have $N$ separate standard model-matching problems~\cite[\S14.5]{ZDG} in~$\tilde Q_{\underline ii}$, with solutions
\begin{align}\label{eq:Qtilde_sol}
	\tilde{\Q}_{\underline i i} = \left[\begin{array}{cc|c}%
A_{\underline{ii}}+\tilde{B}_{2_{\underline{ii}}}\tilde{F}^i & -\tilde{B}_{2_{\underline{ii}}}\tilde{F}^i&0\\
0&A_{\underline{ii}}+(LC)^i&E_{n_{\underline{i}}}^\tp E_{n_{{i}}}L^i\\\hlinet \tilde{F}^i&-\tilde{F}^i&0\end{array}\right]\!,
\end{align}
where $\tilde{F}^i$ and $L^i$ are defined in~(\ref{eq:Riccati}).

The rest of the proof consists of algebraic substitutions and simplifications to transform $\tilde\Q_{\underline i i}$ back into $\K$.

\vspace{1mm}
\noindent\textbf{Step 4.} Applying the inverse transformation from Step~1 to each sub-problem, we obtain $N$ separate control problems in the variables  $\tilde{\Pp}_{\underline{i}i}=\tilde{\Q}_{\underline{i}i}(I+\tilde{\G}_{i\underline{i}}\tilde{\Q}_{\underline{i}i})^{-1}$.

\vspace{1mm}
\noindent\textbf{Step 5.} Given the solutions $\tilde{\Pp}_{\underline ii}$, we can invert the transformation in Lemma~\ref{lem:1} to obtain the adobe-delayed controller $\hat{\Pp}_{\underline ii} = \Pi_{u_i}\tilde{\Pp}_{\underline{i}i}(I-\Pi_{b_i}\tilde{\Pp}_{\underline{i}i})^{-1}$.

\vspace{1mm}
\noindent\textbf{Step 6.} By comparing terms, we have (see shortcut arrow in Fig.~\ref{fig:proof}):
$\Q_{\underline ii} = \Lambda_m^i \hat\Q_{\underline ii} = \Lambda_m^i\hat{\Pp}_{\underline{i}i}(I-{\G}_{i\underline{i}}\Lambda_m^i\Pp_{\underline{i}i})^{-1}$.
Inverting this equation yields 
$\Lambda_m^i\hat{\Pp}_{\underline{i}i}=\Q_{\underline{i}i}(I+{\G}_{i\underline{i}}\Q_{\underline{i}i})^{-1}$. Now zero-pad and horizontally concatenate the $\Lambda_m^i\hat{\Pp}_{\underline{i}i}$ to obtain $\Pp=\sum_{i=1}^N \Lambda_{i} \hat{\Pp}_{:i}E_{p_i}^\tp=\Q(\diag(I+{\G}\Q))^{-1}$. Together with $\Q=\K(I-\G\K)^{-1}$ from Step~1, eliminate $\Q$ to obtain (see also \cite[Lemma~9]{kim2015explicit})
\begin{equation}\label{eq:KPlall}
	\K=\Pp(I+\G (\Pp-\diag( \Pp)))^{-1}.
\end{equation}
Carrying the $\tilde\Q_{\underline i i}$ found in~\eqref{eq:Qtilde_sol} through the transformations in steps 4--6, and simplifying, we obtain
\[
\Pp=-\bar{\1}_m^\tp\bar{\Lambda}\bar{\Pi}_u\tilde{F}(sI-\bar{A}-\bar{L}\bar{C}-\tilde{B}\tilde{F}+\bar{L}\bar{\Pi}_b\tilde{F})^{-1} \bar{L}\bar{\1}_p.
\]
Here, $\bar{\Pi}_u$, $\bar{\Pi}_b$, $\bar\Lambda$, $\tilde F$, $\tilde L$ are block-diagonal matrices (zero-padded if necessary) made from $\{\Pi_{u_i}\}$, $\{\Pi_{b_i}\}$, $\{\Lambda^i_m\}$, $\{\tilde F^i\}$, $\{L^i\}$, respectively, and $\bar{\1}_n\defeq 1_N\otimes I_n$, where $\otimes$ denotes the Kronecker product.
To compute $\K$ and extract subcontrollers $\K_i$, we follow similar steps to the non-delayed case~\cite[Thm.~4]{kashyap2019explicit}.
Substituting the expression for $\Pp$ into~\eqref{eq:KPlall} and further simplifying, we obtain
\begin{multline}\label{eq:Kfinal}
	\K=-\bar{\1}_m^\tp\bar{\Lambda}\bar{\Pi}_u\tilde{F}\bigl(sI-\bar{A}-\bar{L}\bar{C}-\tilde{B}\tilde{F}+\bar{L}\bar{\Pi}_b\tilde{F}\\
	-\bar{L}\bar{\1}_pC_2(sI-A)^{-1}B_2\bar{\1}_{\Delta m}^\tp\bar{\Lambda}\bar{\Pi}_u\tilde{F}\bigr)^{-1} \bar{L}\bar{\1}_p,
\end{multline}
where $\bar{\1}_{\Delta m} \defeq \bar{\1}_m -\diag(\bar{\1}_m)$. Define the controller state variable $\eta$ and input variable $u$ as follows:
\begin{subequations}
	\begin{align*}\notag
		\eta &= -\bigl(sI-\bar{A}-\bar{L}\bar{C}-\tilde{B}\tilde{F}+\bar{L}\bar{\Pi}_b\tilde{F}\\ 
		&\hspace{15mm}-\bar{L}\bar{\1}_pC_2(sI-A)^{-1}B_2\bar{\1}_{\Delta m}^\tp\bar{\Lambda}\bar{\Pi}_u\tilde{F}\bigr)^{-1} \bar{L}\bar{\1}_p y,\label{eta} \\
		u &= \bar{\1}_m^\tp\bar{\Lambda}\bar{\Pi}_u\tilde{F}\eta.
	\end{align*}
\end{subequations}
Further simplify using the identities $\bar{L}\bar{\1}_pC_2=\bar{L}\bar{C}\bar{\1}_n$ and $(sI-A)^{-1}B_2\bar{\1}_{\Delta m}^\tp\bar{\Lambda}=\bar{\1}_{\Delta n}^\tp\bar{\Lambda}(sI-\bar{A})^{-1}\bar{B}$ and split the state equations into their agent-level components $\eta_i$, and we obtain the reduced realization~\eqref{eq:state_agent_Ktau_reduced}, as required. \qedhere

	\bibliographystyle{abbrv}
	{\small\bibliography{optcont_cdc}}
\end{document}